\documentclass[UKenglish,cleveref,autoref,thm-restate,a4paper, authorcolumns]{lipics-v2021} 
\usepackage{mathtools}
\usepackage{cite}
\usepackage{algorithm}
\usepackage{algpseudocode}
\usepackage{xcolor} 
\usepackage[utf8]{inputenc}
\usepackage{url}

\newcommand{\ceil}[1]{\left\lceil #1\right\rceil}

\newcommand{\RE}{\mathbb{R}}            
\newcommand{\bd}{\partial\kern+1pt}     

\newcommand{\SP}{\kern+1pt}             
\newcommand{\PERP}{\kern-2pt \perp \kern-2pt} 
\DeclareMathOperator{\interior}{int}    

\title{On Voronoi diagrams in the Funk Conical Geometry}
\titlerunning{On Voronoi diagrams in the Funk Conical Geometry}

\author{Aditya Acharya}{Department of Computer Science, University of Maryland, College Park, USA \and \url{https://www.cs.umd.edu/~acharya/}}{adach@umd.edu}{https://orcid.org/0000-0002-0359-1913}{}

\author{Auguste H. Gezalyan}{Department of Computer Science, University of Maryland, College Park, USA \and \url{~}}{octavo@umd.edu}{https://orcid.org/0000-0002-5704-312X}{}

\author{David M. Mount}{Department of Computer Science, University of Maryland, College Park, USA \and \url{https://www.cs.umd.edu/~mount/}}{mount@umd.edu}{https://orcid.org/0000-0002-3290-8932}{}

\author{Danesh Sivakumar}{Department of Computer Science, University of Maryland, College Park, USA \and \url{~}}{dsivakumar@terpmail.umd.edu}{https://orcid.org/0009-0008-2484-5549}{}

\authorrunning{A. Acharya, A. H. Gezalyan, and D. M. Mount}

\Copyright{Aditya Acharya, Auguste H. Gezalyan, and David M. Mount}



\date{\today}

\acknowledgements{}

\hideLIPIcs  
\nolinenumbers 
\begin{document}

\maketitle


\begin{abstract}
The forward and reverse Funk weak metrics are fundamental distance functions on convex bodies that serve as the building blocks for the Hilbert and Thompson metrics. In this paper we study Voronoi diagrams under the forward and reverse Funk metrics in polygonal and elliptical cones. We establish several key geometric properties: (1) bisectors consist of a set of rays emanating from the apex of the cone, and (2) Voronoi diagrams in the $d$-dimensional forward (or reverse) Funk metrics are equivalent to additively-weighted Voronoi diagrams in the $(d-1)$-dimensional forward (or reverse) Funk metrics on bounded cross sections of the cone. Leveraging this, we provide an $O\big(n^{\ceil{\frac{d-1}{2}}+1}\big)$ time algorithm for creating these diagrams in $d$-dimensional elliptical cones using a transformation to and from Apollonius diagrams, and an $O(mn\log(n))$ time algorithm for 3-dimensional polygonal cones with $m$ facets via a reduction to abstract Voronoi diagrams. We also provide a complete characterization of when three sites have a circumcenter in 3-dimensional cones. This is one of the first algorithmic studies of the Funk metrics.
\end{abstract}

\section{Introduction}
The Funk metric was originally defined by Paul Funk in 1929 to study geometries with straight line geodesics \cite{funk1929geometrien} (see Section~\ref{sec:prelims} for definitions). This work directly addresses Hilbert's fourth problem, which sought to characterize all geometries where straight lines are geodesics. The forward and reverse Funk metrics are deeply connected to the Hilbert metric, introduced by David Hilbert in 1895 \cite{hilbert1895linie}. The Hilbert metric is a generalization of the Cayley-Klein metric of hyperbolic geometry to general convex bodies. In fact, the Hilbert metric can be defined as the average of the forward and reverse Funk metrics. The Funk metric also plays a role in the definition of other distance functions, such as the Thompson metric \cite{thompson1963certain}, which is defined as the maximum of the forward and reverse Funk metrics.

As the Hilbert geometry has many applications in information geometry \cite{nielsen2019clustering, nielsen2022nonlinear, karwowski2025bicone}, convex approximation \cite{abdelkader2018delone, abdelkader2024convex, vernicos2018flag}, and real analysis \cite{lemmens2013birkhoff} (among other fields), it has become of interest to develop classical algorithms of computational geometry such as Voronoi diagrams \cite{gezalyan2023voronoi, bumpus2023software}, Farthest-site Voronoi diagrams \cite{song2025farthest}, Delaunay triangulations \cite{gezalyan2024delaunay}, and support vector machines \cite{hilbertSVM,acharya2026classifiers} in the Hilbert geometry. Software for experimenting in the forward and reverse Funk metric is available \cite{banerjee2025software}. This motivates the study of the Funk metrics in two ways. First, understanding the behavior of the Funk geometry provides insights for constructing improved algorithms in related metrics, such as Hilbert and Thompson. Second, there is evidence that the Funk metrics are more natural to study than Hilbert. Faifman argues this perspective in ``A Funk perspective on billiards, projective geometry and Mahler volume'' \cite{faifman2024funk}. Faifman shows that the Holmes-Thompson volume and surface area are preserved under polarity in the Funk metrics with a myriad of other desirable properties. 

The Funk metrics have also recently emerged as an area of contemporary research in Finsler geometries. Funk is considered one of the most basic Finsler geometries \cite{papadopoulos2014handbook}. In particular, Papadopoulos and Troyanov showed that every convex body contains a canonical weak Finsler structure whose distance is inherently the forward Funk metric \cite{papadopoulos2008weak}. Additionally, Shen used it to construct an R-flat spray and employed the forward Funk metric to find a Finsler metric with zero curvature \cite{shen2001funk}. Most recently, the forward Funk was used to solve problems involving trajectory in the context of Zermelo's navigation problem \cite{chavez2025geometry}.

Cones are a particularly natural setting to study the Funk and Hilbert metrics. For instance, the cone of symmetric positive definite matrices is a fundamental setting for information geometry \cite{nielsen2023fisher}. The study of the Hilbert metric (the average of the forward and reverse Funk metrics) on cones is a well-established area of real analysis, including the study of Perron-Frobenius operators, Bellman operators in Markov decision processes and Shapley operators in stochastic games, which are all operators on cones \cite{lemmens2013birkhoff, karwowski2025bicone,metz1995hilbert}. Furthermore, the Funk metric has also been studied in the context of the Holmes-Thompson metric and  makes essential use of conical geometry \cite{arya2026cauchy}. As bounded convex domains are simply cross sections of cones, our results extend implicitly to bounded convex domains \cite{nielsen2019clustering}.

\textbf{Our Contributions}. We study Voronoi diagrams in both the forward and reverse Funk metrics in cones. Voronoi diagrams are a natural first structure to study, as they give insights into the behavior of the metric and underlie many other geometric algorithms, such as Delaunay triangulations and nearest neighbor searching. We establish several key properties in arbitrary dimension including: (1) bisectors are composed of a set of rays emanating from the apex of the cone, and (2) Voronoi diagrams in the $d$-dimensional Funk metrics are equivalent to additively-weighted Voronoi diagrams in the $(d-1)$-dimensional Funk metrics on bounded cross sections of the cone. Leveraging this, we provide an $O\big( n^{\ceil{\frac{d-1}{2}}+1}\big)$ algorithm for creating these diagrams in $d$-dimensional elliptical cones using a transformation to and from Apollonius diagrams, and an $O(mn\log(n))$ algorithm for $3$-dimensional polygonal cones with $m$ facets via a reduction to abstract Voronoi diagrams. We also provide a complete characterization of when three sites admit a circumcenter in a $3$-dimensional cone. This and a concurrent work \cite{arya2026cauchy} are the first algorithmic studies of the Funk metrics.

\section{Preliminaries} \label{sec:prelims}

In this paper, when we refer to a cone $C\subset \RE^d$ we mean a fully dimensional convex set which is closed under positive scaling (that is, whenever $p\in C$ so is $\beta p$ for any $\beta \geq 0$). We assume our cones are \emph{pointed}, meaning that they contain no lines. This implies that each cone has a unique \emph{apex} at the origin. We denote the boundary of $C$ by $\bd C$ and its interior by $\interior C$. We will examine two types of cones, elliptical and polygonal cones (see Figure \ref{fig:ConeTypes}). For a point $p \in \RE^d$, we write $C_p$ for the cone $C$ translated so that its apex is at $p$ and $C^r_p$ for the cone $-C$ (obtained by reflecting $C$) translated so that its apex is at $p$. We denote the Euclidean distance by $d_2(\cdot,\cdot)$. 

\begin{figure}[h]
    \centering
    \includegraphics[width=.6\textwidth]{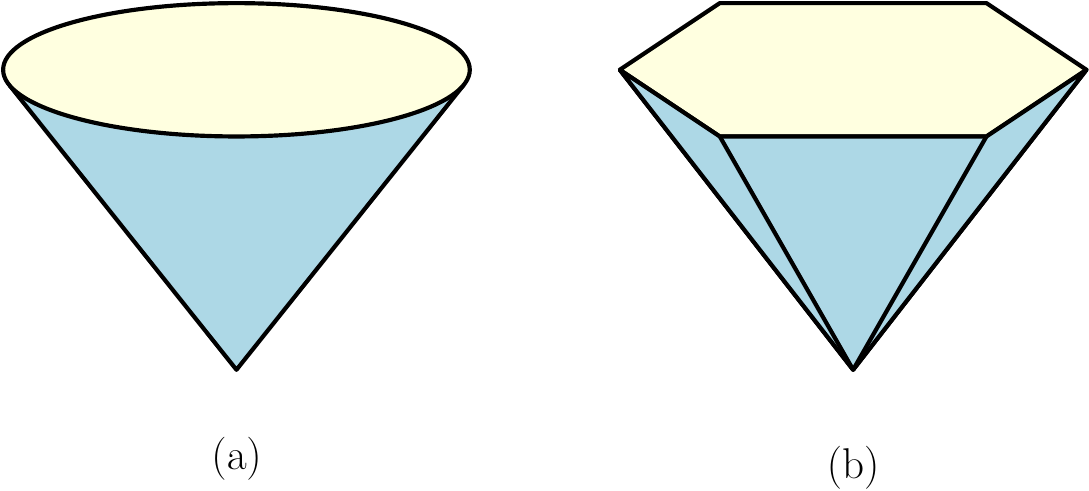}
    \caption{(a) An elliptical cone. (b) A polygonal cone.}
    \label{fig:ConeTypes}
\end{figure}

 Next, let us recall the definition of a metric space and introduce the forward and reverse Funk metrics. In this section we provide a summary of well known results \cite{papadopoulos2014handbook,lemmens2013birkhoff} with proofs for completeness.

\begin{definition}[Metric Space]
\label{def:metric}
The pair $(\mathbf{X},d)$, where $\mathbf{X}$ is a set and $d:\mathbf{X}\times \mathbf{X} \rightarrow \RE$, is a \textbf{metric space} if for all $ a,b,c \in \mathbf{X}$:
\begin{enumerate}
    \item $d(a,b) \geq 0$, $d(a,b) = 0$ if and only if $a=b$
    \item $d(a,b) = d(b,a)$
    \item $d(a,c)\leq d(a,b)+d(b,c)$
\end{enumerate}
\end{definition}

Metric spaces generalize the notion of a distance on a set. When some of the three properties fail we call $d$ a \emph{weak metric}. To define the forward and reverse Funk metrics, we begin by defining a partial order on $C$.

\begin{definition}[Conical Partial Order]
    Given a convex cone $C$, we define the following partial order on $\interior C$: for $a, b \in \interior C$, we say $a\leq_C b$ if and only if $b-a \in C$. Equivalently, $a\leq_C b$ if and only if $b \in C_a$ (see Figure~\ref{fig:HasseDef} in the Appendix).
\end{definition}

We define the forward and reverse Funk metrics in a cone through this partial order.

\begin{definition}[Forward Funk metric] \label{def:infCone}
For $a,b \in \interior C$, the forward Funk metric $F_C  (a,b)$ is
\[
    F_{C}(a,b)
        ~ = ~  \ln \inf \{\beta \in \RE \mid a \leq_C \beta b\}.
\]
\end{definition}

\begin{definition}[Reverse Funk metric] \label{def:revinfCone}
For $a,b \in \interior C$, the reverse Funk metric is \[F^r_C(a,b) = F_C(b,a).\]
\end{definition}

Note that these definitions generalize the standard definition of the Funk metric in a bounded convex domain to a cone in a way that is consistent with the Hilbert metric definition on the cone \cite{lemmens2013birkhoff}. In particular, for any $a,b \in \Omega$ where $\Omega$ is a bounded cross section of $C$, the restriction satisfies $F_\Omega (a,b)=F_C(a,b)$. The Funk metrics are weak metrics on $C$ (see Definition \ref{def:metric}) that satisfy neither the positivity property (property 1) nor the symmetry property (property 2). They only satisfy the triangle inequality (property 3, see Lemma \ref{lem:triangleInequality}).

In bounded convex domains, the Funk metric is always positive. In unbounded regions, however, numerators and denominators in the standard Funk formulation (see Remark 2.3 \cite{lemmens2013birkhoff}) can be infinite. The cone-based definitions above (see Definitions \ref{def:infCone} and \ref{def:revinfCone}), which follow Birkhoff's version of the Hilbert metric \cite{lemmens2013birkhoff} handle this naturally. A consequence of these definitions is that the Funk metrics become signed distance functions and can take negative values. When $b$ is in $C_a$, the value of $\ln(\beta)$ is negative; when $b$ is on the boundary of $C_a$, $\ln(\beta)$ is $0$; and when $b$ is below $C_a$, $\ln(\beta)$ is positive (see Figure \ref{fig:FunkDistance} (a)-(c)). For completeness, we prove straight lines are geodesics and that the triangle inequality holds in the Funk metrics in the Appendix (Lemmas \ref{lem:triangleInequality} and \ref{lem:lineGeodesic}).

\begin{figure}[h]
    \centering
    \includegraphics[width=1.0\textwidth]{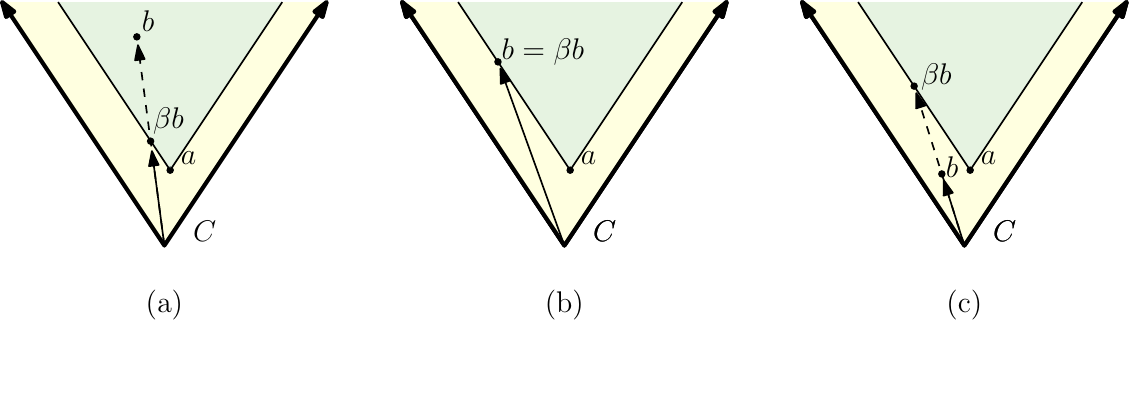}
    \caption{Two points $a$ and $b$ where (a) $F_C(a,b)<0$, (b) $F_C(a,b)=0$, (c) $F_C(a,b)>0$.}
    \label{fig:FunkDistance}
\end{figure}

Since the infimum definition of the forward Funk metric is not directly amenable to computation, we provide an alternative. To compute $F_C(a, b)$, it suffices to find the $\beta$ such that $\beta b$ intersects $\bd C_a$ (see Figure \ref{fig:FunkDistance}). This geometric observation leads to an equivalent and more computationally feasible definition using similar triangles. 

\begin{proposition} \label{prop:TriangleCone}
Let $a, b \in \interior C$. Consider the $2$-dimensional slice through $a$, $b$, and the origin. Let $R$ be the boundary ray in this slice such that the ray from the origin through $b$ intersects $\bd C_a$ on a side parallel to $R$. Then, $F_{C}(a,b) = \ln \frac{d_2(a, R)}{d_2(b, R)}$, where $d_2(a,R)$ and $d_2(b,R)$ denote the perpendicular Euclidean distances from $a$ and $b$ to $R$ respectively.
\end{proposition}

\begin{proof}
    In the 2-dimensional slice let $\beta b$ be where the ray through the origin and $b$ intersects $\bd C_a$. Since this boundary ray is parallel to $R$ we can see that $d_2(\beta b, R)=d_2(a,R)$. By similar triangles we can see that $\frac{d_2(\beta b, R)}{d_2(b,R)}=\beta$. Hence  $\frac{d_2(a, R)}{d_2(b,R)}=\beta$. Taking the natural log of both sides gives the result (see Figure \ref{fig:similarTriangles}(a)-(c)).
\end{proof}

 Recall that the ball of radius $\rho$ around $p$ is $B_\rho(p) = \{q : d_2(p,q) \leq \rho\}$. We now characterize balls in the forward and reverse Funk metrics. We refer to $B_\rho ^F(a)$ and $B_\rho^{rF}(a)$ as the forward and reverse Funk balls of radius $\rho$ around a point $a$.

\begin{lemma}\label{lem:ForwardBallCone}
Given $a \in \interior C$ and $\rho \in \RE$, the forward and reverse Funk balls satisfy:
\begin{enumerate}
\item[$(i)$] $B_\rho^F(a) = C_{e^{-\rho}a}$
\item[$(ii)$] $B_\rho^{rF}(a) = C^r_{e^{\rho}a}$.
\end{enumerate}
\end{lemma}    

\begin{proof}We show that $p \in B_\rho^F(a)$ if and only if $p \in C_{e^{-\rho}a}$. Consider the 2D slice of the cone given by the plane through $a$, $p$, and the origin. By the similar triangles formula, $F_C(a,p) = \ln \frac{d_2(a,R)}{d_2(p,R)}$, where $R$ is the appropriate ray of $C$. Since scaling a point by $e^{-\rho}$ scales its distance to $R$ by the same factor, we have $d_2(e^{-\rho}a, R) = e^{-\rho}d_2(a,R)$. Therefore,
\[
F_C(a,p) \leq \rho \iff \frac{d_2(a,R)}{d_2(p,R)} \leq e^\rho \iff \frac{e^{-\rho}d_2(a,R)}{d_2(p,R)} \leq 1 \iff d_2(e^{-\rho}a, R) \leq d_2(p,R).
\]
This holds if and only if $p \in C_{e^{-\rho}a}$. The reverse Funk is analogous.\end{proof}

\begin{figure}[h]
    \centering
    \includegraphics[width=.9\textwidth]{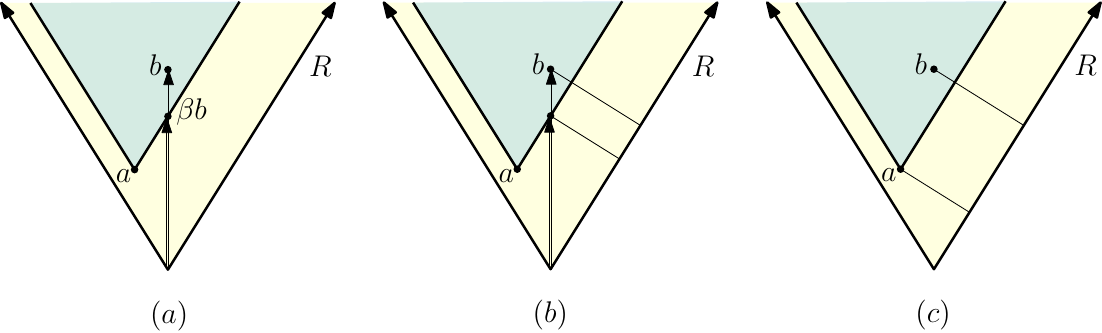}
    \caption{Computing $F_C(a,b)$: (a) scaling $b$ to $\bd C$, (b) using a ratio of distances to the boundary, (c) using the distance of $a$ to the boundary.}
    \label{fig:similarTriangles}
\end{figure}
\section{Funk Conical Voronoi Diagrams}

Let $S=\{s_1,\dots,s_n\}\subset \interior    C$ be a set of sites. Given a metric $\zeta$ the Voronoi cell of a site $s$ is:\[V(s)=\{q\in C\mid \zeta(s,q)\leq \zeta(s',q) \text{ for all }s' \in S\}\]
When sites are weighted $S=\{(s_1,w_1),\dots,(s_n,w_n)\}\subset \interior C$ the Voronoi cell of a site $s$ is: 
\[V((s,w))=\{q\in C\mid \zeta(s,q)+w\leq \zeta(s',q) +w'\text{ for all }(s',w') \in S\}\]
The bisector between two sites $J(s,s')$ is the set of points equidistant to both.
Since the Funk metrics are asymmetric, we define the following conventions:
\begin{enumerate}
    \item The \emph{forward Funk} Voronoi cell of $s$ is  $\{q\in C\mid F_C(s,q)\leq F_C(s',q) \text{ for all }s' \in S\}$
    \item The \emph{reverse Funk} Voronoi cell of $s$ is  $\{q\in C\mid F^r_C(s,q)\leq F^r_C(s',q) \text{ for all }s' \in S\}$
\end{enumerate}
Weighted versions for both are defined analogously using the weighted Voronoi cell definition.

\subsection{Empty Voronoi Cells} 

We begin by studying Voronoi diagrams in the Funk metrics in arbitrary convex cones. We make the general position assumption that no site lies on the boundary of a cone centered at another site. We first characterize when a site has an empty Voronoi cell. In the next section, we show that forward and reverse Funk Voronoi cells are star-shaped. We assume the cone is convex with apex at the origin.

\begin{lemma}\label{lem:ForwardDominate}
    Given sites $a, b \in \interior C \subset \RE^d$, if $b \in \interior  C_a$ then for all $c \in \interior C$, $F_C(a,c) < F_C(b,c)$ (see Figure~\ref{fig:FunkDistance}(a)). In this case we say that $a$ \textbf{dominates} $b$ in the Funk metric.
\end{lemma}

\begin{proof}
 Recall from Definition~\ref{def:infCone}, 
\[
    F_C(a,c) 
        ~ = ~ \ln \inf\{\beta \in \RE \mid \beta c \in C_a\} \text{ and } F_C(b,c) 
        ~ = ~ \ln \inf\{\beta \in \RE \mid \beta c \in C_b\}.
\]
Since $b \in \interior C_a$, we have $C_b \subsetneq C_a$ so every $\beta$ satisfying $\beta c \in C_b$ also satisfies $\beta c \in C_a$. Therefore, the infimum over the larger set $\{\beta : \beta c \in C_a\}$ is smaller: $F_C(a,c) < F_C(b,c)$.
\end{proof}

From this lemma, we obtain the following immediate consequence.

\begin{corollary}
    Let $S$ be a set of sites in $C$. If $a \leq_C b$ for some $a, b \in S$, then $b$ has an empty Voronoi cell in the Voronoi diagram of $S$ in the forward Funk.
\end{corollary}

\begin{lemma}\label{lem:BackwardDominate}
Given two sites $a, b \in \interior C \subset \RE^d$, if $b \in \interior C^r_a$ then for all $c \in \interior C$, $F^r_C(a,c) < F^r_C(b,c)$. In this case we say that $a$ \textbf{dominates} $b$ in the reverse Funk metric.
\end{lemma}

\begin{proof}
Since $b \in \interior C^r_a$, $a \in C_b$. Hence, $a \geq_C b$. Now consider a third point $c \in \interior C$. Let $\eta$ be the infimum value such that $c \leq_C \eta b$, that is, $\eta b - c \in C$. Algebraically, we can see that
\[ 
    \eta a - c 
        ~ = ~ \eta a +\eta b - \eta b - c 
        ~ = ~ \eta (a - b) + (\eta b -c).
\] 
Since both $\eta(a-b)$ (as $a-b\in \interior C$) and $\eta b -c$ are in $C$ we know $\eta a - c \in \interior C$. Hence $\eta a \in C_c$. So: $\inf\{\beta : c \leq_C \beta a\} < \inf\{\beta : c \leq_C \beta b\}.$ By taking logarithms: $F_C(c,a) < F_C(c,b)$, and $F_C^r(a,c) < F_C^r(b,c)$.
\end{proof}

\begin{corollary}
    Let $S$ be a set of sites in $C$. If $b \leq_C a$ for some $a, b \in S$, then $b$ has an empty Voronoi cell in the Voronoi diagram of $S$ in the reverse Funk.
\end{corollary}

\subsection{Voronoi Cells}

In this section, and future sections, we will denote the Voronoi cell of a site $p$ as $V(p)$. We prove that Voronoi cells are star-shaped. Recall that a set is star-shaped if there exists a point $p$ in it such that for all $q$ in the set, the line segment from $p$ to $q$ is contained in the set. 

\begin{lemma} \label{lem:star}
The Voronoi cell of a weighted site $(p,w_p)$ in the forward and reverse Funk metrics in a cone $C$ is star-shaped with respect to $p$ (see Figure \ref{fig:FunkCells} (b)).
\end{lemma}

\begin{proof}
Let us consider the forward Funk. Suppose towards a contradiction that there exists a site $(p,w_p)$ and points $x, y \in \interior C$ such that $x \in V(p)$, $y \notin V(p)$, and $y$ lies on the line segment $p$ to $x$. Since $x$ and $y$ lie on the same ray from $p$, with $y$ between $p$ and $x$, by the geodesic property of straight lines (Lemma~\ref{lem:lineGeodesic}):
\[
    F_C(p, x) 
        ~ = ~ F_C(p, y) + F_C(y, x).
\]
Let $(q,w_q)$ be the site closest to $y$. Since $y \notin V(p)$, we have $F_C(q, y)+w_q < F_C(p, y)+w_p$. Since $x \in V(p)$, we have $F_C(p, x) +w_p\leq F_C(q, x)+w_q$. Combining these inequalities:
\[
    F_C(q, x) + w_q 
        ~ \geq ~ F_C(p, x) + w_p
        ~ = ~ F_C(p, y) + F_C(y, x) + w_p
        ~ > ~ F_C(q, y) + w_q + F_C(y, x).
\]
This contradicts the triangle inequality for the Funk metric. Therefore, Voronoi cells are star-shaped with respect to their sites. This applies equally across bounded cross sections. The proof is analogous for the reverse Funk.
\end{proof}

\begin{figure}[h]
    \centering
    \includegraphics[width=.7\textwidth]{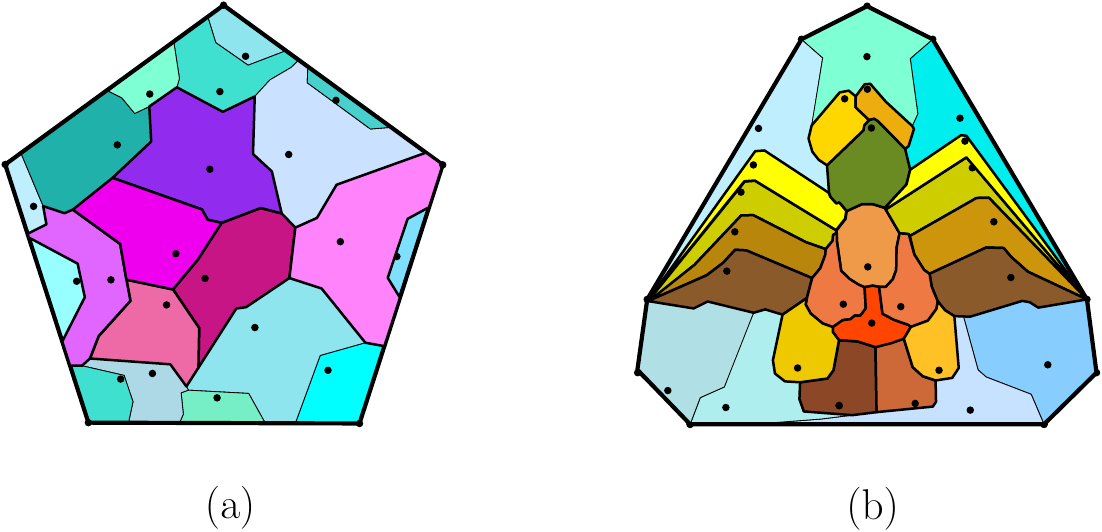}
    \caption{(a) Voronoi cells in the reverse Funk metric, and (b) in the forward Funk metric.}
    \label{fig:FunkCells}
\end{figure}

\subsection{Voronoi Algorithms}

\subsubsection{Bisectors as Rays Through the Origin}

We now proceed to discuss bisectors. We start with the additive properties of the Funk metrics and use them to show that bisectors in the Funk metrics are made out of rays through the apex of the cone (the origin).
\begin{lemma}\label{lem:additive}
Let $p, q \in \interior C$ and $\lambda > 0$. Then the forward and reverse Funk metrics satisfy the following additive decomposition along rays from the apex of the cone: 

\begin{enumerate}
    \item $F_C(p, q) 
        = F_C(p, \lambda p) + F_C(\lambda p, q) = -\ln(\lambda) + F_C(\lambda p, q)$

    \item $F^r_C(p, q) = F^r_C(p, \lambda p) + F^r_C(\lambda p, q) = \ln(\lambda) + F^r_C(\lambda p, q).$
\end{enumerate}
\end{lemma}
\begin{proof}
Consider the 2D plane through $p$, $q$, and the origin. Let $R$ be the boundary ray of $C$ in this plane such that the forward Funk distance from $p$ to $q$ is computed using $R$ (via the similar triangles formula). Since $p$ and $\lambda p$ lie on the same ray from the origin, all three points $p$, $\lambda p$, and $q$ lie in this same plane, and we can compute all distances using the same boundary ray $R$ (by similar triangles the distance between $p$ and $\lambda p$ could be calculated with respect to any ray through the apex). We have:
\begin{align*}
    F_C(p, \lambda p) + F_C(\lambda p, q) 
        & ~ = ~ \ln \frac{d_2(p, R)}{d_2(\lambda p, R)} + \ln \frac{d_2(\lambda p, R)}{d_2(q, R)}
          ~ = ~ \ln \frac{d_2(p, R)}{d_2(\lambda p, R)} \frac{d_2(\lambda p, R)}{d_2(q, R)} \\
        & ~ = ~ \ln \frac{d_2(p, R)}{d_2(q, R)} 
          ~ = ~ F_C(p, q).
\end{align*}
Note that $\ln \frac{d_2(p, R)}{d_2(\lambda p, R)} = -\ln(\lambda)$. For the reverse Funk we recall that $F^r_C(a,b) = F_C(b,a)$. Using a similar process as before:
\begin{align*}
    F^r_C(p, q) 
        & ~ = ~ F_C(q, p)
          ~ = ~ F_C(q, \lambda p) + F_C(\lambda p, p) \\
        & ~ = ~ F^r_C(\lambda p, q) + F^r_C(p, \lambda p)
          ~ = ~ F^r_C(p, \lambda p) + F^r_C(\lambda p, q).
\end{align*}The rest follows from: $F^r_C(p,\lambda p) = F_C(\lambda p,p) = \ln(\lambda)$.
\end{proof}

Adding these provides an alternative proof of the Hilbert metric's projective invariance. Next we show that bisectors in the Funk metrics are made of rays through the origin. 

\begin{theorem}\label{thrm:RaysThroughApex}
The bisector between two sites $a, b \in \interior C$ in the forward and reverse Funk metrics consists of a set of rays emanating from the origin.
\end{theorem} 

\begin{proof}
To show that bisectors consist of a set of rays emanating from the origin, we will prove that if $F_C(a, p) = F_C(b, p)$, then $F_C(a, \lambda p) = F_C(b, \lambda p)$ for all $\lambda > 0$. Since $p$ and $\lambda p$ lie on the same ray through the origin, distances from any fixed site decompose additively along this ray (by the same reasoning as Lemma~\ref{lem:additive}). Thus:
\[
    F_C(a, p) 
        ~ = ~ F_C(a, \lambda p) + F_C(\lambda p, p) 
        \text{~~~and~~~} 
    F_C(b, p) 
        ~ = ~ F_C(b, \lambda p) + F_C(\lambda p, p).
\]
Since $F_C(a, p) = F_C(b, p)$ by hypothesis, we have 
\[
    F_C(a, \lambda p) + F_C(\lambda p, p) 
        ~ = ~ F_C(b, \lambda p) + F_C(\lambda p, p).
\]
Subtracting $F_C(\lambda p, p)$ from both sides yields $F_C(a, \lambda p) = F_C(b, \lambda p)$.
Since for any point $p$ on the bisector, $F_C(a, p) = F_C(b, p)$, we can see that the entire ray from the origin through $p$ lies on the bisector.

For the reverse Funk it can analogously be shown that if $F^r_C(a, p) = F^r_C(b, p)$, then $F^r_C(a, \lambda p) = F^r_C(b, \lambda p)$ for all $\lambda > 0$. 
\end{proof}

\begin{theorem}\label{thrm:reduction}
Calculating the forward and reverse Funk Voronoi diagrams of a set of sites in a cone $C \subset \RE^d$ is equivalent to computing the forward and reverse Funk weighted Voronoi diagrams on any $(d-1)$-dimensional bounded cross-sectional convex region $\Omega$.
\end{theorem}

\begin{proof}
Since bisectors are composed of rays through the origin (see Theorem \ref{thrm:RaysThroughApex}), the Voronoi diagram on a bounded cross section $\Omega$ of the cone is equivalent to calculating the Voronoi diagram on the entire cone.  Sites are projected through the origin onto $\Omega$ and given additive weights determined by their scaling factor (see Lemma \ref{lem:additive}). A point $p$ belongs to the Voronoi cell of a site $x$ in $C$ if and only if its projection onto $\Omega$ belongs to the weighted Voronoi cell of the projected site (a consequence of Lemma \ref{lem:additive}). 
\end{proof}

\subsubsection{Bisectors in 2-Dimensional Cones}

We provide explicit equations for bisectors in $2$-dimensional cones. This will provide geometric intuition for the higher-dimensional constructions.

\begin{lemma}\label{lem:forward2DBisector}
    Given two sites $p, q \in \interior C \subset \RE^2$ such that neither dominates the other:
    \begin{enumerate}
        \item The \emph{forward Funk bisector} is a ray from the origin through $\bd C_p \cap \bd C_q$. 
        \item The \emph{reverse Funk} bisector is a ray from the origin through $\bd C^r_p \cap \bd C^r_q$.
    \end{enumerate}
\end{lemma}
\begin{proof}
    Consider the forward Funk. Let the boundary rays of $C$ be $R_1$ and $R_2$, given by the equations $R_1: c_1 x + c_2 y = 0$ and $R_2: c_3 x + c_4 y  = 0$. Let $p = (x_1, y_1)$ and $q = (x_2, y_2)$ lie in angular order $\langle R_1, p, q, R_2 \rangle$. This ordering determines that distances from $p$ are calculated using $R_1$, while distances from $q$ are calculated using $R_2$. 
    
    We define the linear forms $L_1(x,y) = c_1 x + c_2 y $ and $L_2(x,y) = c_3 x + c_4 y $ corresponding to $R_1$ and $R_2$ respectively. Note that $L_i(x,y)$ gives a value proportional to the signed distance from $(x,y)$ to ray $R_i$. A point $(x,y)$ lies on the Funk bisector if and only if $F_C(p, (x,y)) = F_C(q, (x,y))$. By the similar triangles formula, this equality becomes (we can avoid absolute values because of the angular order and exponentiate to remove logarithms): 
    \[\frac{L_1(x_2, y_2)}{L_1(x,y)} = \frac{L_2(x_1, y_1)}{L_2(x,y)}.\] Cross-multiplying yields: \[L_1(x_2, y_2) \cdot L_2(x,y) = L_2(x_1, y_1) \cdot L_1(x,y).\] Or equivalently: \[L_2(x,y) \cdot L_1(x_2, y_2) - L_1(x,y) \cdot L_2(x_1, y_1) = 0.\]
    
    Since this equation is linear in $x$ and $y$, the bisector is a line. This line passes through both the origin (the apex of $C$), since the linear forms have no constant term, and the intersection point $\bd C_p \cap \bd C_q$, since it is the intersection of the zero radius Funk balls.

    For reverse Funk the situation is analogous with $C^r$ in place of $C$ and the fractions inverted. The resulting formula is as follows:
    \[L_2(x, y)\cdot L_1(x_1,y_1) - L_1(x, y)\cdot L_2(x_2,y_2)=0. \qedhere\]  
\end{proof}

Note that Voronoi diagrams in both 2-dimensional Funk geometries are equivalent to angular sorting after pruning dominated sites.

\subsection{Voronoi Diagrams in d-Dimensional Elliptical Cones }

We now extend our results to $d$-dimensional elliptical cones. Since the \emph{Funk metrics are affine invariants}, we will consider the circular case, to which the elliptical case can be affinely mapped  \cite{papadopoulos2014handbook}. We show that we can compute the forward and reverse Funk $d$-dimensional Voronoi diagrams in the circular cones by computing the Apollonius diagram on a set of disks representing the zero balls of the sites in $(d-1)$-dimensions with additional linear time filtering and relocation. \emph{We focus on the forward Funk in this section, but the reverse Funk is analogous except when noted}.

Consider a $d$-dimensional circular cone $C \subset \RE^d$ with apex at the origin and a set of sites $S=\{s_1,\dots,s_n\}$ in the interior of $C$. By Theorem~\ref{thrm:RaysThroughApex} the bisectors in the Funk metrics consist of rays through the origin. 

To begin, we consider the sites by Euclidean distance from the apex and let $\Omega$ be an orthogonal cut to the cone's axis below the lowest site for the reverse Funk and above the highest for the forward Funk. Note that $\Omega$ is a $(d-1)$-dimensional ball.

\textbf{Constructing Apollonius Sites:} For each site $s_i\in S$, let $B_i=C_{s_i}\cap \Omega$, the ball made from intersecting the Funk zero radius ball of $s_i$ with $\Omega$. For each $s_i$ we define the corresponding Apollonius site $(c_i,w_i)$ where $c_i$ is the Euclidean center of ball $B_i$ with Euclidean radius $w_i$. Note the Euclidean center and radius are different from the Funk center and radius. Denote this new set of sites in the Apollonius as $S'$.

\textbf{Structural Equivalence:} We now establish that the combinatorial structure of the forward and reverse Funk Voronoi diagrams are determined by the Apollonius diagram.

\begin{lemma}\label{lem:tangentBalls}
    If $p$ lies on the bisector of two sites $s_i$ and $s_j$ in the forward Funk Voronoi diagram, then the ball around $p$ of radius $F_C^r(p,s_i)=F_C^r(p,s_j)$ is tangent to $B_i$ and $B_j$.
\end{lemma}

\begin{proof}
    This follows from the fact that $C_{F^r(p,s_i) p}$'s apex is at $\bd C_{s_i}\cap \bd C_{s_j}$ and therefore must trace the outside of the balls where they touch (see Figure \ref{fig:Tangency} (a)).
\end{proof}

\begin{corollary}
    If a point $p$ is on the bisector of two sites $s_i$ and $s_j$ in the reverse Funk Voronoi then the ball around $p$ with radius $F_C(p,s_i)=F_C(p,s_j)$ is tangent to $B_i$ and $B_j$.
\end{corollary}

\begin{proof}
    This follows as Lemma \ref{lem:tangentBalls}, except the cones are reflected (see Figure \ref{fig:Tangency} (b)).
\end{proof}

\begin{figure}[h]
    \centering
    \includegraphics[width=.7\textwidth]{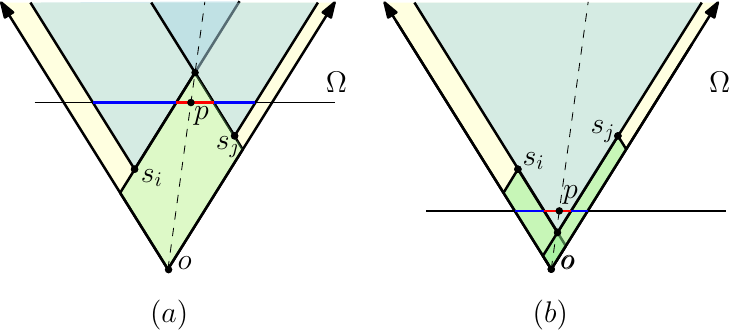}
    \caption{(a) Reverse Funk ball tangent to two forward Funk balls, (b) forward Funk ball tangent to two reverse Funk balls.}
    \label{fig:Tangency}
\end{figure}

\begin{lemma}\label{lem:filterReverse}
Let $v'$ be a vertex in the Apollonius Voronoi diagram determined by sites  $(c_1,w_1),\dots,(c_{d},w_{d})$,
and let $B$ be the corresponding ball tangent to the $d$ Apollonius sites. Then $v'$ corresponds to a vertex in reverse Funk if and only if $B\subset \Omega$.
\end{lemma}

\begin{proof}
If $B$ is not in $\Omega$ then it does not correspond to a forward Funk ball around a vertex $v$ and hence it cannot be equidistant to the sites. 

If $B$ is in $\Omega$ then consider the cone $C'$ congruent to $C$ whose cross section with $\Omega$ is $B$. Let its apex be $v$. Then the zero ball of $v$ is exactly $B$ and since $B$ is tangent to the zero balls of $s_1,s_2,\dots,s_d$, we know $v$ is equidistant to those sites and $v$ is a reverse Funk vertex.
\end{proof}

\begin{lemma}\label{lem:vertex-existence}
If $d$ sites $s_1,\dots,s_d \in S$ determine a Voronoi vertex in the forward Funk metric in $\Omega$, then the corresponding sites in the Apollonius $(c_1,w_1),\dots,(c_d,w_d)$ determine a vertex in the Apollonius diagram.
\end{lemma}

\begin{proof}
    A point $p$ is a vertex in the forward Funk metric between $d$ sites $s_1,s_2,\dots,s_d$ if and only if the reverse Funk ball around $p$ with corresponding radius is tangent to the zero radius balls around $s_1,s_2,\dots,s_d$. Call this ball $B$. Note that $B$ is therefore tangent to $(c_1,w_1),\dots,(c_d,w_d)$, and hence the Euclidean center of $B$ is equidistant in the Apollonius.
\end{proof}

\begin{figure}[h]
    \centering
    \includegraphics[width=.8\textwidth]{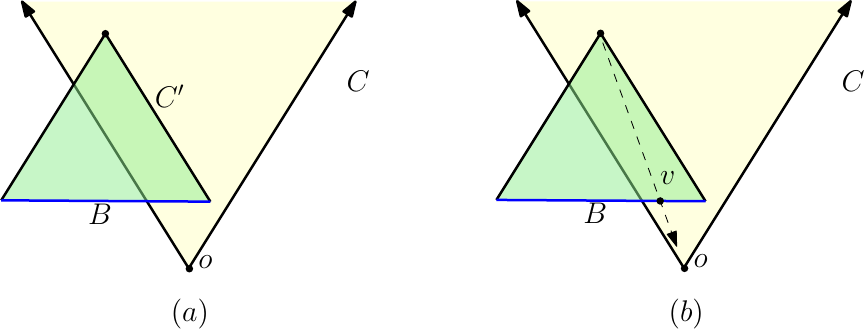}
    \caption{(a) The apex of $C'$ lies in $C$. (b) Locating the reverse Funk vertex $v$ in $B$.}
    \label{fig:Tangency2}
\end{figure}

\begin{lemma}\label{lem:filterForward}
Let $v'$ be a vertex in the Apollonius Voronoi diagram determined by sites  $(c_1,w_1),(c_2,w_2),\dots,(c_d,w_d)$,
and let $B$ be the corresponding ball tangent to the $d$ Apollonius sites. Then $v'$ corresponds to a vertex in the forward Funk Voronoi diagram.
\end{lemma}

\begin{proof}
Let $v'$ be an Apollonius vertex with corresponding ball $B$. Since $B$ is tangent to $d$ zero balls contained in $\Omega$, the ball $B$ intersects $\Omega$. Consider the cone through $B$, $C'$, that is congruent to $C^r$. Note that for any point $p$ on $\bd C'\cap \Omega$ the boundary ray of $C$ through $p$ extending upwards must lie entirely inside $C$ (see Figure \ref{fig:Tangency2} (a)). Since the ray emanates from the apex of $C'$ the apex of $C'$ must be in $C$. Hence $C'$ corresponds to the zero ball around some point in $C$ which must therefore be a Voronoi vertex in the forward Funk Voronoi.
\end{proof}

\begin{corollary}
    In the reverse Funk metric, $d$ points have a Voronoi vertex only if they have one in the forward Funk metric. 
\end{corollary}

\begin{proof}
    Apollonius vertices are reverse Funk Voronoi vertices only if their forward Funk balls are subsets of $\Omega$ as opposed to just intersecting $\Omega$ (see Lemma \ref{lem:filterForward}).\end{proof}

\textbf{Algorithm.} Let $\Omega$ be a bounded cut above the sites in forward Funk or below for reverse Funk. We compute the forward Funk or reverse Funk Voronoi diagram as follows:

\begin{enumerate}
    \item \textbf{Setup:} For each site $s_i \in S$, compute the zero ball $B_i = C_{s_i} \cap \Omega$. Extract the Euclidean center $c_i$ and radius $w_i$ of $B_i$ to form the Apollonius site $(c_i, w_i)$. Let $S'$ denote the resulting set of Apollonius sites.
    
    \item \textbf{Apollonius:} Compute the Apollonius diagram of $S'$ in $(d-1)$ dimensions with complexity $O(n \log h)$ (where $h$ is the number of nondominated sites) when $d-1=2$ \cite{karavelas2002dynamic} or $O(n^{\ceil{\frac{d-1} {2}}+1})$ when $d-1>2$ \cite{boissonnat2005convex}. Let $V'$ denote the set of Apollonius vertices.
    
    \item \textbf{Relocate:} For each Apollonius vertex $v' \in V'$, we compute the corresponding Funk vertex $v$ by doing the following: Consider the corresponding ball $B$ and the cone congruent to $C^r$ ($C$ in reverse Funk) passing through it. Calculate the cone's height and position a point at its apex. Then consider the ray through that point to the apex of $C$. Where that ray intersects $B$ is where the Funk vertex must be. In the forward Funk it always exists, but in the reverse it only exists if $B \subset \Omega$ (see Lemmas \ref{lem:vertex-existence} and \ref{lem:filterReverse}) (see Figure \ref{fig:Tangency2} (b)). 
\end{enumerate}

\begin{theorem} \label{thm:runtimeEllispe}
The forward and reverse Funk Voronoi diagrams of $n$ sites in a $d$-dimensional circular cone can be computed in $O(T_{d-1}(n))$ time, where $T_k(n)$ is the time to compute an Apollonius diagram of $n$ sites in $k$ dimensions.
\end{theorem}

\begin{proof}
 In our setup we intersect planes and cones, which is purely algebraic and $O(1)$ per site, then we prepare our data, which is constant time per site so $O(n)$. The Apollonius diagram is computed in time $O(T_{d-1}(n))$. Note there are at most $O(T_{d-1}(n))$ vertices \cite{boissonnat2003combinatorial}. For each of these we require only $O(1)$ to relocate through a few simple algebraic steps. 
\end{proof}

\begin{corollary}
    The Voronoi diagram of a set of sites in an elliptical cone in the Funk metrics can be computed in time $O(n \log h)$ when $d-1=2$, and $O(n^{\ceil{\frac{d-1}{2}}+1})$ otherwise.
\end{corollary}

\begin{proof}
    This follows directly from Theorem \ref{thm:runtimeEllispe} and existing Apollonius algorithms \cite{karavelas2002dynamic,boissonnat2005convex}.
\end{proof}

\subsection{Voronoi Diagrams in 3-Dimensional Polygonal Cones.}

For polygonal cones we first characterize the properties of the forward and reverse Funk Voronoi diagram, then we describe how to compute some primitives and show that it can be solved with the abstract Voronoi diagram framework. As we did before, we consider the sites by Euclidean distance from the apex. Let $\Omega$ be an orthogonal cut to the cone's axis below the lowest site for reverse Funk and above the highest for forward Funk. Note that $\Omega$ is a 2-dimensional polygon. Projecting our sites to $\Omega$ gives them weights. To make the Voronoi diagram we will use a modified version of the concept of spokes from the Hilbert metric \cite{nielsen2017balls}. 

\begin{definition}[Forward Spokes]
For a point $p$ inside a polygon $\Omega$ with vertices $v_1, \ldots, v_m$, the \emph{forward spokes} from $p$ are the rays from $p$ to $v_i$. These spokes partition $\Omega$ into sectors around $p$ (see Figure \ref{fig:Spokes} (a)).
\end{definition}
\begin{definition}[Reverse Spokes]
For a point $p$ inside a polygon $\Omega$ with vertices $v_1, \ldots, v_m$, the \emph{reverse spokes} from $p$ are the rays from $p$ in the direction away from $v_i$. These spokes partition $\Omega$ into sectors around $p$ (see Figure \ref{fig:Spokes} (b)).
\end{definition}
\begin{figure}[h]
    \centering
    \includegraphics[width=.7\textwidth]{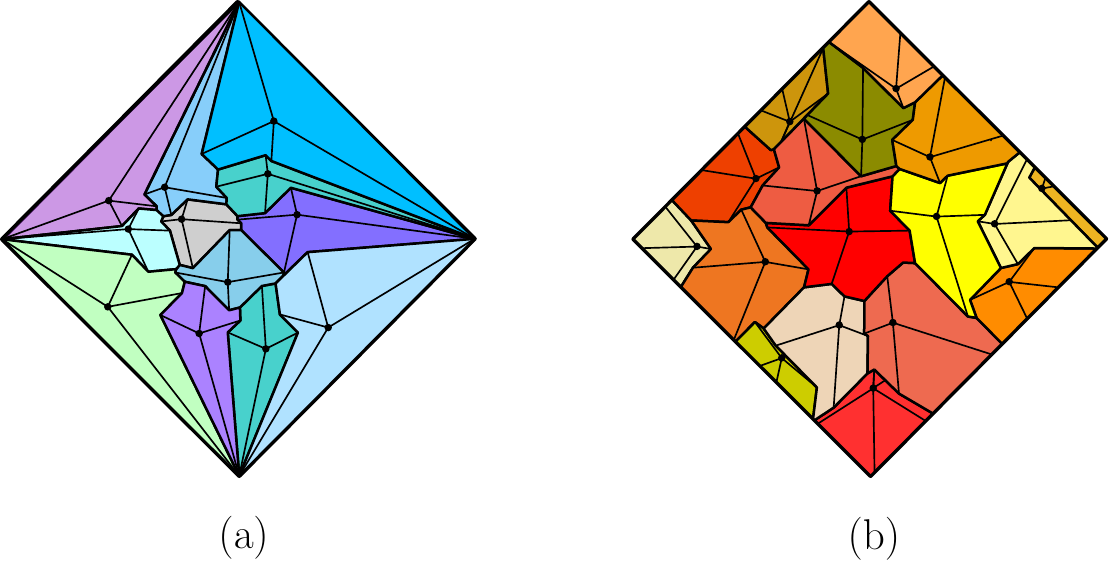}
    \caption{(a) Spokes in a forward Funk Voronoi and, (b) reverse Funk Voronoi.}
    \label{fig:Spokes}
\end{figure}
\begin{lemma}
Given two weighted sites $(s_1,w_1),(s_2,w_2)$ in a polygon $\Omega$, the forward and reverse Funk metric bisectors between them are piecewise composed of line segments separated by sector.
\end{lemma}

\begin{proof}
To start, we consider the forward Funk. The spokes from $s_1$ and $s_2$ partition $\Omega$ into finitely many regions. Within each region, the forward Funk distances to points from both $s_1$ and $s_2$ are computed using fixed boundary edges $E_i$ and $E_j$ of $\Omega$. Let the sites and the edges lie in order $
\langle E_i,s_1,s_2,E_j \rangle$. Let $E_i$ and $E_j$ be given by the linear forms $E_i(x,y) = 0$ and $E_j(x,y) = 0$. A point $(x,y)$ lies on the weighted bisector if $F_C(s_1, (x,y)) + w_1 = F_C(s_2, (x,y)) + w_2$. By the similar triangles formula:

\[ \ln\left(\frac{E_j(s_1)}{E_j(x,y)} \right)+w_1= \ln\left(\frac{E_i(s_2)}{E_i(x,y)}\right)+w_2.\]
Which becomes:
\[E_j(s_1) \cdot E_i(x,y) = e^{w_2 - w_1} \cdot E_i(s_2) \cdot E_j(x,y).\]

Which is linear in $(x,y)$. Therefore, the bisector in each sector is a line segment (or empty if it does not intersect the sector). By the continuity of $F_C(s_1, (x,y)) + w_1 - F_C(s_2, (x,y)) - w_2$ the bisector must be one continuous curve. So we have a piecewise linear bisector overall.

In the reverse Funk the proof is analogous except we use reverse spokes and get:
\[ \ln\left(\frac{E_i(x,y)}{E_i(s_1)} \right)+w_1= \ln\left(\frac{E_j(x,y)}{E_j(s_2)}\right)+w_2.\]
Which becomes:
\[E_j(s_2) \cdot E_i(x,y) = e^{w_2 - w_1} \cdot E_i(s_1) \cdot E_j(x,y). \qedhere \]
\end{proof}

We claim that the bisectors between two sites are composed of at most $m$ sides. This will allow us to bound the complexity of our algorithm.

\begin{lemma}\label{lem:polygonalBisector}
Given two weighted sites $(s_1,w_1), (s_2,w_2)$ in a polygon $\Omega$ with $m$ vertices, the forward and reverse Funk bisectors consist of $O(m)$ line segments.
\end{lemma} 

\begin{proof}Each of our sites has $m$ spokes coming out of it (going towards vertices for forward Funk and away from vertices for reverse Funk). These intersect to form sectors. In each sector the bisectors are line segments that are composed piecewise to make a connected curve. However, once a bisector passes over a spoke it can never pass back over it without violating the star-shapedness of Voronoi cells.

Since each site contributes at most $m$ spokes, the bisector crosses at most $2m$ spoke boundaries, and changes equation at most $2m$ times. So the bisector has complexity $O(m)$. 
\end{proof}

Now we can solve for the Voronoi diagram with the abstract Voronoi diagram framework \cite{klein1993randomized,klein1988abstract,mehlhorn1991construction, klein2009abstract}. In particular we will draw from ``Abstract Voronoi diagram revisited'' \cite{klein2009abstract}. First we will show that the forward and reverse Funk metrics satisfy the required axioms.

\textbf{Axiom A1 (Jordan Curves):} By Lemma \ref{lem:polygonalBisector}, 
    each bisector consists of $O(m)$ line segments forming a piecewise-linear curve in $\Omega$. For every bisector we extend where they hit $\bd \Omega$ toward infinity. The resulting curve is piecewise linear in $\mathbb{R}^2$, unbounded, and projects to a Jordan curve through the north pole under stereographic projection.

\textbf{Axiom A2 (Path-connected regions):} By Lemma \ref{lem:star} inside $\Omega$, each Voronoi cell is star-shaped with respect to its site and hence path-connected. Then bisectors are extended as rays outside of $\Omega$, the portion of the cell outside $\Omega$ becomes bounded by half-planes and the boundary of $\Omega$, and hence is path-connected. Since the interior portion is path connected, and the exterior is path-connected, and they share a boundary, their union is path connected.

\textbf{Axiom A3 (Coverage):} Inside $\Omega$, every point is at finite 
weighted Funk distance from all sites, hence belongs to the closure of at 
least one Voronoi region. Outside of $\Omega$, the extended bisector rays 
partition $\mathbb{R}^2 - \Omega$ into regions, each with an associated 
site. Therefore, the closures of all Voronoi regions cover $\mathbb{R}^2$.

To get the exact running time we must now solve the necessary primitive for our chosen abstract Voronoi diagram framework \cite{klein2009abstract}. This requires us to solve for the Voronoi diagram of five sites as a primitive. Then ``Abstract Voronoi diagram revisited''\cite{klein2009abstract} gives our runtime.

\begin{lemma}
    The Funk Voronoi diagrams of five sites can be computed in $O(m)$ time.
\end{lemma}

\begin{proof}
We apply a modified version of the randomized incremental algorithm from previous work on Voronoi diagrams in the Hilbert metric \cite{gezalyan2023voronoi} but on five sites. We calculate the Funk distance between all pairs of sites and remove dominated ones. The algorithm 
inserts sites one at a time, tracing the boundary of each new Voronoi cell with spokes (we only need forward spokes or reverse spokes depending on the Funk metric)
through bisector segments and sector boundaries. For a constant number of 
sites (five in this case), the diagram has $O(m)$ complexity (refer to Section~\ref{sec:circumcenter} to see that three sites contribute at most two Voronoi vertices) with $O(m)$ total 
bisector segments and construction time $O(m)$.\end{proof}

\begin{theorem}
The forward or reverse Funk Voronoi diagram of $n$ sites in a 3-dimensional polygonal cone with $m$ facets can be computed in $O(nm \log (n))$ time.
\end{theorem}

\subsection{Circumcenter Characterization in 3-Dimensional Cones}\label{sec:circumcenter}

One particular aspect of the Funk metrics is that three points in the Funk metrics do not necessarily lie on the boundary of a ball. The characterization with respect to both Funk metrics is almost entirely the same. We will make our definitions direction agnostic. Given three sites we characterize when they admit a forward or reverse Funk circumcenter in 3-dimensional cones. This analysis applies to any convex cone. 

Let $p,q,r$ be three sites in $C$. Without loss of generality assume that $r$ is the site furthest from the apex of the cone (or closest in the reverse Funk). We assume that neither $p$ dominates $q$ nor $q$ dominates $p$. Let $H$ be a plane through $r$ that is perpendicular to the axis of the cone. Let $\Omega = H\cap C$. By Theorem \ref{thrm:reduction} we obtain weighted sites $p'$ with weight $w_p$ 
and $q'$ with weight $w_q$ on $\Omega$. Since $r$ lies on the plane, it 
requires no weight.

\begin{lemma}\label{lem:ballInPlane}
    If $p,q,r$ admit a circumcenter $u$ in $C$ in the forward or reverse Funk metrics, then, every site along the ray through the apex and $u$ is the center of a ball through $p,q,r$.
\end{lemma}
\begin{proof}
This follows from Theorem \ref{thrm:RaysThroughApex}.\end{proof}

Hence we can determine circumcenter existence in planar cross sections of $C$ with weighted point-sites rather than the full cone. We introduce our first region. 

\begin{definition}[Inner Regions: $Z_0$]
Define the zero balls $Z_p = C_p \cap \Omega$ and $Z_q = C_q \cap \Omega$, and let $Z_0 = Z_p \cup Z_q$.
\end{definition}

These are the weighted zero radius balls of $p$ and $q$ in the plane.

 \begin{lemma}\label{lem:Z0}
    If $r\in \interior Z_0$ then $r$ has no cell in the Voronoi diagram.
\end{lemma}
\begin{proof}
    If $r$ is in $Z_p$ then $r$ is dominated by $p$ (see Lemma \ref{lem:BackwardDominate} and Lemma \ref{lem:ForwardDominate}).
\end{proof}

For the remainder of the section assume $r \not \in Z_0$. The weighted bisector $J((p',w_p), (q',w_q))$
intersects $\partial \Omega$ at two points $u$ and $v$ in either Funk metric (by Lemma \ref{lem:star}). We now define the tangent balls at these points (equivalent to infinite balls in Hilbert \cite{gezalyan2024delaunay}).

\begin{definition}[Tangent Balls]
Let $u, v \in \partial \Omega$ be the two points where $J((p',w_p), (q',w_q))$ intersects $\partial \Omega$ such that $p', u, q'$ is in clockwise order around $\Omega$, and $p', v, q'$ is in counter-clockwise order. Define $B(p:q)$ to be the Funk ball centered at $u$ tangent to both $Z_p$ and $Z_q$, and $B(q:p)$ to be the Funk ball centered at $v$ tangent to both. 
\end{definition}

Both these Funk balls must be in the direction opposite of the bisector's incident direction. We define the overlap region to be the intersection of the two tangent balls (see Figure \ref{fig:Regions} (a)).

\begin{definition}[Overlap Region]
Given two sites $p, q \in \interior C$, the overlap region in the Funk metrics is the intersection of both tangent balls at $\bd \Omega, $ $Z_{1}(p,q) = B(p:q) \cap B(q:p).$
\end{definition}

\begin{lemma}\label{lem:Z1}
    If $r$ is in $Z_1(p,q)$  then $p,q,r$ have no circumcenter. 
\end{lemma} 

\begin{proof}
The bisector $J((p',w_p), (q',w_q))$ parameterizes centers of balls tangent to $Z_p$ and $Z_q$ within $\Omega$. The balls $B(p:q)$ and $B(q:p)$ are the limiting tangent balls as the center approaches $\partial \Omega$. Any ball passing through a point in $Z_1 = B(p:q) \cap B(q:p)$ would have to have its center outside of $\Omega$. Hence $r$ lies interior to all tangent balls to $Z_p$ and $Z_q$.
\end{proof}

Next we define the one center regions, then two center regions. 

\begin{definition}[One-Center Region]
Define $B_{1}(p,q) = (B(p:q) \cup B(q:p)) - Z_1$.
\end{definition}
\begin{lemma}\label{lem:B1}
    If $r$ lies within $B_1(p,q)$ then there is one ball through $p,q$ and $r$.
\end{lemma}

\begin{proof}
Without loss of generality let $r \in B(p:q)$ and $r \not \in B(q:p)$. Then the pencil of Funk balls from $v$ to $u$ passes from one side of $r$ to the other. By the intermediate value theorem, there is a Funk ball passing through $p,q,r$ in $\Omega$. By continuity, it passes through once.
\end{proof}

\begin{definition}[Two Center Region]
Given two sites $p, q \in \interior C$, the two center region is the space between $Z_p,Z_q,B(p:q),B(q:p)$. We call it $Z_2(p,q)$.
\end{definition}

\begin{lemma}\label{lem:Z2}
    If $r\in Z_2(p,q)$ then it lies on the boundary of two Funk balls through $p$ and $q$.
\end{lemma}

\begin{proof}
    This region is the region traced out by the boundaries of the Funk balls (in the opposite direction to the Voronoi) outside of $B(p:q)$ and $B(q:p)$. This region can come in three forms: there is a $Z_1$ region and this $Z_2$ region is split into two pieces on either side of it, there is no $Z_1$ region and the $Z_2$ region is one connected piece, there is no $Z_1$ region and $Z_p$ and $Z_q$ intersect and the $Z_2$ is split above and below them (see Figure \ref{fig:Regions} (b)-(e)). We will tackle all these cases at once. Consider that by the nature of $r$ not being in either infinite ball then $p$ and $q$ are both closer to the boundary than $r$ so $r$'s cell is sandwiched between them. This means there are at least two circumcenters. However by star-shapedness of the cells we can see that we can leave and enter the cell at most once each and so there are exactly two circumcenters.
\end{proof}

\begin{figure}[h]
    \centering
    \includegraphics[width=.8\textwidth]{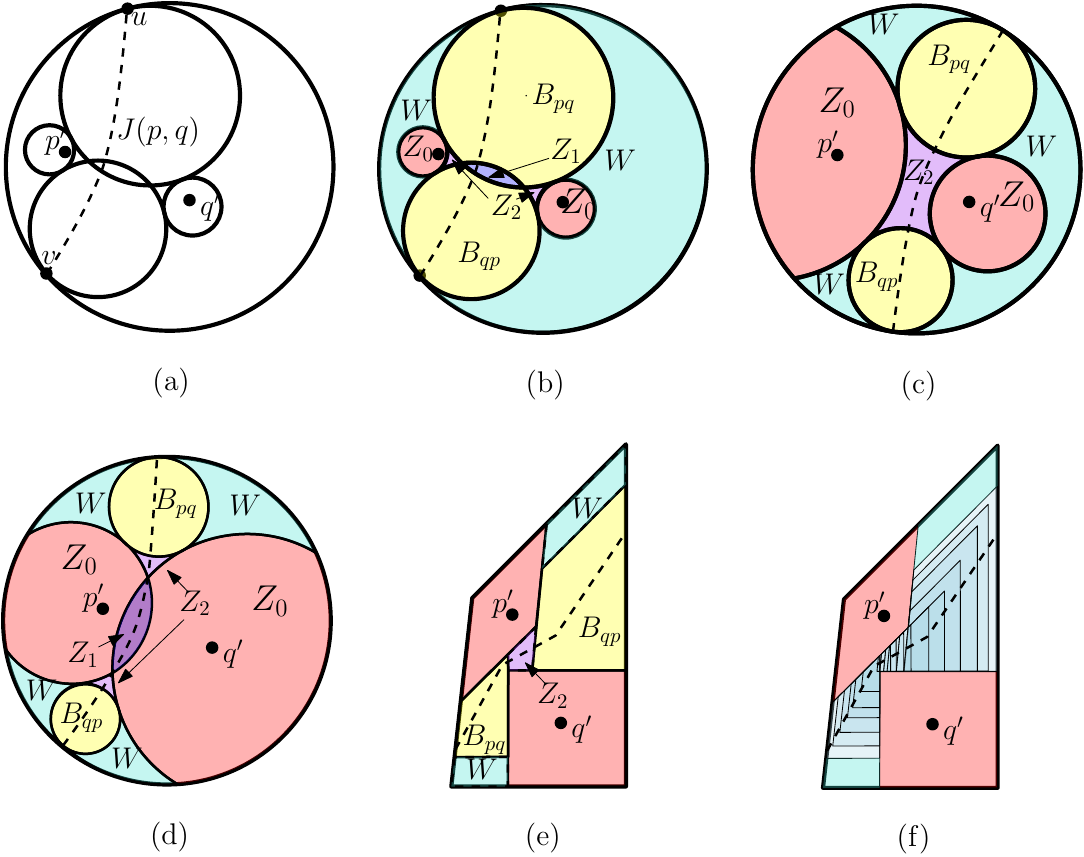}
    \caption{In the reverse Funk, (a) Cross section of a circular cone, (b) (c) (d) regions labeled in a circular cone with differently weighted sites, (e) regions labeled in a polygonal cone, (f) pencil of balls along the bisector in a polygonal cone.}
    \label{fig:Regions}
\end{figure}

Now we define our final region, the outer region. This is the area exterior to all the other regions, on the outside of $B(p:q)$, $B(q:p)$, $Z_p$, and $Z_q$ (see Figures \ref{fig:Regions}(b)-(e)).

\begin{definition}[Outer Region]
   The outer region of $p$ and $q$, $W(p,q)$ is $\Omega - (Z_0(p,q) \cup Z_1(p,q) \cup Z_2(p,q) \cup B(p:q) \cup B(q:p))$.
\end{definition}

\begin{lemma}\label{lem:W}
    If $r\in W(p,q)$ then there is no ball through $r,p,$ and $q$.
\end{lemma}

\begin{proof}
    In this case $r$ is outside of the two infinite balls $B(p:q)$ and $B(q:p)$ and $Z_2(p,q)$. Hence no ball on the pencil of balls on the bisector will pass through it (see Figure \ref{fig:Regions} (f)).
\end{proof}

\begin{theorem}\label{thm:circumcenter-characterization}
Let $p, q$ be two sites in a 3-dimensional cone $C$ with $r \in C$ such that neither 
$p$ nor $q$ dominates each other in the forward Funk metric (or reverse analogously). Then the number of forward
Funk circumcenters (or reverse) through $p, q, r$ is determined by which region $r$ 
occupies:
\begin{itemize}
    \item If $r \in Z_0(p,q)$, there are 0 circumcenters (Lemma \ref{lem:Z0}).
    \item If $r \in Z_1(p,q)$, there are 0 circumcenters (Lemma \ref{lem:Z1}).
    \item If $r \in B_1(p,q)$, there is exactly 1 circumcenter (Lemma \ref{lem:B1}).
    \item If $r \in Z_2(p,q)$, there are exactly 2 circumcenters (Lemma \ref{lem:Z2}).
    \item If $r \in W(p,q)$, there are 0 circumcenters (Lemma \ref{lem:W}).
\end{itemize}
\end{theorem}

\section{Conclusions}

In this paper, we presented one of the first algorithmic studies of the Funk metric in the conical geometry. Our main contributions include: establishing that bisectors consist of rays through the apex of the cone; proving a dimensional reduction from $d$-dimensional Funk Voronoi diagrams to $(d-1)$-dimensional additively weighted Voronoi diagrams on bounded cross sections of the cone; providing efficient algorithms for computing these in $d$-dimensional elliptical cones and 3-dimensional polygonal cones; and categorizing circumcenters for 3-dimensional cones. Several natural extensions remain open: extending our polygonal result to arbitrary dimension, investigating the relationship between Funk Voronoi and Hilbert Voronoi, and studying other classical computational geometry structures in the Funk metric.

\bibliography{shortcuts,hilbert}

\section{Appendix}
\begin{figure}[h]
    \centering
    \includegraphics[width=1.0\textwidth]{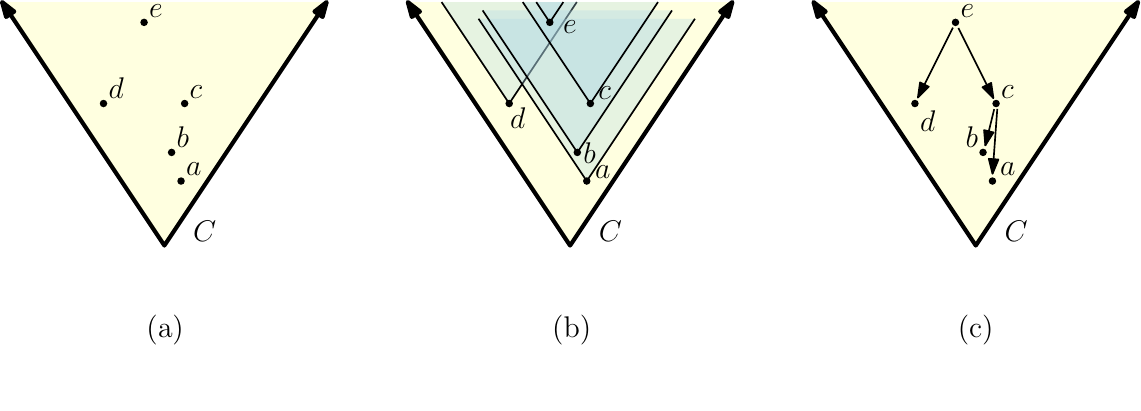}
    \caption{(a) A collection of 5 points in a cone $C$. (b) The cone $C$ translated to the points. (c) A directed Hasse diagram describing their partial order, $d \leq_C e, c \leq_C e, b \leq_C c, a \leq_C b$. }
    \label{fig:HasseDef}
\end{figure}

Here we prove for completeness that the triangle inequality holds in the Funk metrics and that geodesics are straight lines. We also provide a figure that illustrates the nature of the partial order of a set of points in a cone under the Funk metric (see Figure \ref{fig:HasseDef}).

\begin{lemma}\label{lem:triangleInequality}
Given $a, b, c\in \interior  C$, the forward and reverse Funk metrics satisfy the triangle inequality: 
\begin{enumerate}
    \item $F_{C}(a, c) 
        ~ \leq ~ F_C(a, b) + F_{C}(b, c)$
    \item $F^r_{C}(a, c) 
        ~ \leq ~ F^r_C(a, b) + F^r_{C}(b, c)$
\end{enumerate}

\end{lemma}
\begin{proof} 
Let $\lambda$ be the infimum value so that  $\lambda b$ is in $C_a$, and let $\eta$  be the infimum value so that $\eta c \in C_b$. By the transitivity of $\leq_C$ we can see that $\lambda \eta c \in C_a$. Hence, $e^{F_C(a,c)} \leq \lambda \eta = e^{F_C(a,b)} e^{F_C(b,c)}$ and so by taking logarithms $F_{C}(a, c) \leq F_C(a, b) + F_{C}(b, c)$. The reverse Funk metric follows directly from definition:
\[
    F_C^r(a,c)
        ~ = ~ F_C(c,a)
        ~ \leq ~ F_C(c,b) + F_C(b,a)
        ~ = ~ F_C^r(a,b) + F_C^r(b,c). \qedhere 
\]
\end{proof}

\begin{lemma}\label{lem:lineGeodesic}
Straight lines are geodesic in the forward and reverse Funk metrics.
\end{lemma}

\begin{proof}
We first consider the forward Funk metric. Consider three collinear points $a, b, c \in \interior C$ where $b$ lies on the line segment $a$ to $c$.  The ray from $a$ through $b$ is the same as the ray from $a$ through $c$, so both distances $F_C(a,b)$ and $F_C(a,c)$ are computed using the same boundary ray $R$ of $C$. Writing out the distances, we have
\begin{align*}
    F_C(a,b) + F_C(b,c) 
        & ~ = ~ \ln\frac{d_2(a,R)}{d_2(b,R)} + \ln\frac{d_2(b,R)}{d_2(c,R)} \\
        & ~ = ~ \ln\frac{d_2(a,R)}{d_2(c,R)} 
          ~ = ~ F_C(a,c). 
\end{align*}
The reverse Funk is analogous.
\end{proof}
Note that straight lines are not necessarily the only geodesics in the Funk metric.

\end{document}